\documentclass[journal]{IEEEtran}
\IEEEoverridecommandlockouts
\usepackage{amsmath,graphicx,amsthm,amssymb,epsfig,graphicx,amssymb,url,bm}
\usepackage{threeparttable}
\usepackage{cases}
\usepackage{latexsym}
\usepackage{lscape}
\usepackage{arydshln}
\usepackage{balance}
\usepackage{color}
\providecommand{\keywords}[1]{\textbf{\textit{Index terms---}} #1}

\newtheorem{prop}{Proposition}

\newtheorem{lemma}{Lemma}

\theoremstyle{definition}

\DeclareMathOperator*{\argmax}{argmax}

\begin{document}

\bibliographystyle{IEEEtran}
\graphicspath{{figures/}}

\title{ Energy Efficient Cooperative Network Coding with Joint Relay Scheduling and Power Allocation}
\author{
        Nan Qi, 
        Ming Xiao, \IEEEmembership{Senior Member,~IEEE,}
        Theodoros A. Tsiftsis, ~\IEEEmembership{Senior Member,~IEEE,}\\
        Mikael Skoglund, \IEEEmembership{Senior Member,~IEEE,}
       Phuong L. Cao,~\IEEEmembership{Student Member,~IEEE,}
        and Lixin Li,~\IEEEmembership{Member,~IEEE}
\thanks{
Nan Qi and Lixin Li are with Northwestern Polytechnical University, China (e-mail: naq@kth.se, lilixin@nwpu.edu.cn).

Ming Xiao, Mikael Skoglund, Phuong L. Cao are with the School of Electrical Engineering of KTH, Royal Institute of Technology, Stockholm, Sweden (e-mail: \{mingx, skoglund, plcao\}@kth.se). Nan Qi is also with KTH.

Theodoros A. Tsiftsis is with the School of Engineering, Nazarbayev University, Kazakhstan (e-mail:  theodoros.tsiftsis@nu.edu.kz).

This work was supported in part by  Fundamental Research Funds for the Central Universities (No.3102014JCQ01052), EU Marie Curie Project, QUICK, the China Scholarship Council (CSC),  China Postdoctoral Science Foundation (No.2014M552489), Postdoctoral research project funded by Shaanxi Province, Natural Science Basic Research Plan in Shaanxi Province of China (No.2016JM6062), Aerospace Science and Technology Innovation Fund of China Aerospace Science and Technology Corporation and the Doctorate Foundation of Northwestern Polytechnical University (CX201317).
}

}
\maketitle
\vspace{-1.5cm}
\begin{abstract}
The energy efficiency (EE) of a multi-user multi-relay  system with the maximum diversity network coding (MDNC) is studied. We explicitly find the connection among the outage probability, energy consumption and EE and formulate the maximizing EE problem under the outage probability constraints. Relay scheduling (RS) and power allocation (PA) are applied to schedule the relay states (transmitting,
sleeping, \emph{etc}) and optimize the transmitting power under the practical channel and power consumption models. Since the optimization problem is NP-hard, to reduce computational complexity, the outage probability is first tightly approximated to a log-convex form. Further, the EE is converted into a subtractive form based on the fractional programming. Then a convex mixed-integer nonlinear problem (MINLP) is eventually obtained. With a generalized outer approximation (GOA) algorithm, RS and PA are solved in an  iterative manner. The Pareto-optimal curves between the EE and the target outage probability show the EE gains from PA and RS. Moreover, by comparing with the no network coding (NoNC) scenario, we conclude that with the same number of relays,
MDNC can lead to EE gains.  However, if RS is implemented, NoNC can outperform
MDNC in terms of the EE when more relays are needed in the MDNC scheme.

\end{abstract}

\keywords{Energy efficiency-outage probability tradeoff, network coding, relay scheduling, power allocation, generalized outer approximation (GOA).}

\IEEEpeerreviewmaketitle

\section{Introduction}
\subsection{Related Works and Motivation}

Energy management for the energy constrained networks continues to be a challenging issue and has attracted considerable research interest recently (see \cite{3} and references therein). To assess how efficiently the energy is managed, energy efficiency (EE) has been proposed and it is defined as the sum of successfully transmitted bits per unit energy \cite{5}. As a single metric, it conveniently combines the quality-of-service parameter (e.g., throughput, bit error rate and outage probability) with energy consumption together and continues to be a popular topic in recent years \cite{5}, \cite{10}, \cite{9}-\cite{13}.

Network coding (NC) in the physical layer is a commonly used technique to improve the energy efficiency. With NC, the messages of different users are combined into new codewords at the intermediate nodes \cite{10}-\cite{18}. Recently, there are a significant amount of works on the physical layer network coding for the multi-access networks. More specifically, based on binary network coding, classical butterfly network and the two-way relay scheme  have been widely studied \cite{10}-\cite{binL}. Relay selection was conducted to maximize the energy efficiency (EE) \cite{10}, to minimize the bit error rate \cite{binT}, \cite{binY} and to maximize the throughput \cite{binM}, \cite{binL}. To enable more users to cooperate, analog network coding (ANC) was applied. In \cite{ANCM}, the packet flow and the channel occupation frames were both  scheduled to improve the network throughput. In \cite{25s}, the EE-maximization problem for ANC-based two-way relay (TWR)  networks has been tackled. In \cite{analognc}, the authors have shown that when the number of sources/relays or the modulation order increases, ANC may be more energy efficient than the conventional time-orthogonal non-cooperative  transmission scheme (referred to as No Network Coding scenario (NoNC)). However, above network coding schemes were suboptimal in terms of reliability, especially in high signal-to-noise ratio (SNR) regions resulting in a suboptimal diversity order. To achieve the full diversity order for a group of cooperative users, the idea of the maximum diversity network coding (MDNC) has been proposed in \cite{19} and \cite{20}. It was shown that, for both orthogonal and non-orthogonal channels, an $M$-user $N$-relay network based on MDNC can achieve a diversity of $N-M+1$ if direct source-BS channels are absent, while the  analog network coding \cite{25s} cannot in general \cite{19}. It was also proved in \cite{19} that MDNC can provide the network with a larger outage capacity than ANC in the high SNR region. Further, given fixed data transmitting power and noise power, reference \cite{17} illustrated that an increasing Galois field size can bring energy saving. {In \cite{19} and \cite{20}, the authors studied only the diversity order (i.e., the exponent of SNR in the upper bound) of the MDNC networks}.

Apart from energy efficient network coding, resource allocation strategies to improve the EE mainly include the following: maximizing rates without increasing the energy costs \cite{10}, \cite{binT}, minimizing the outage probability or energy costs to indirectly improve the EE \cite{9}, \cite{14} and optimizing energy aware infrastructure including schedule the modes of relays and base stations (BSs) to decrease the electrical energy cost \cite{14}, \cite{25}-\cite{energysaving2}. In particular, to maximize the EE, there is a tradeoff between decreasing the electrical circuit energy consumption by utilizing fewer relays and decreasing the outage probability along with transmission power by employing more relays to improve the diversity order. This naturally raises the problem of \emph{optimum relay selection} problem. In particular, there are many articles focusing on relay selection rules  \cite{JKKKJL}, \cite{14}, \cite{AS2}-\cite{25} and they differ in network setups (single- \cite{JKKKJL}, \cite{bobai}, \cite{25}, \cite{13} vs. multi-relay and single source-destination pair \cite{14}), relay protocol (with GF(2) \cite{binT}, \cite{binY}, \cite{25s}, ANC \cite{analognc} vs. without network coding (e.g., Amplify-and-Forward (AF) or Decode-and-Forward (DF)) \cite{JKKKJL}, \cite{14}, \cite{bobai}-\cite{13}) and channel conditions (CSI available \cite{VSK}-\cite{13} vs. unavailable at the transmitters \cite{binM}, \cite{14}, \cite{VSK}, \cite{JMJGCE} ).
Different relay selection methods have been adopted for the networks with or without CSIT. For the network with CSIT, the relay was selected according to the average CSI including the path-loss, transmission distance \cite{AS2}, instantaneous SNR \cite{25} or fading states of links \cite{14}. For the network without CSIT, a (max-)min-max criterion to select one single relay such that the BER \cite{binY} or the outage probability \cite{binM}, \cite{VSK}, \cite{JMJGCE} of the worst user/channel was optimized. To maximize the throughput, reference \cite{binT} proposed a relay selection criterion that decides according to the weighted rate sum for any bidirectional rate pair on the boundary of the achievable rate region individually. {Relay selection strategies in \cite{14}-\cite{25} are only valid for the networks where a selected relay is only connected with one unique source-destination pair. In their studied systems, the source message can be recovered as long as any relay manages to forward its message to the destination. However, in the networks that are coded over Galois Field (mostly nonbinary), one relay is connected with multiple sources. Additionally, the destination must obtain enough number of codewords  from a group of selected relays to jointly recover multiple source messages; otherwise, an outage event happens.} 
 {Thus, the derivation of the outage probability is different from those in published cooperative relaying protocols}, including binary network coding \cite{binT}, \cite{binY}, \cite{25s}, analog network coding \cite{analognc} and no network coding \cite{JKKKJL}, \cite{14}, \cite{25}.
As mentioned above, {since one relay is connected with multiple users, the ``quality" of one specific relay can not be measured by just two channels, i.e., user-relay and relay-destination channels. Additionally, in contrast to the published relaying protocols in \cite{14}-\cite{25}, relays work cooperatively rather than separately in our MDNC scheme.  We have to select an optimum relay group such that the messages of a specific user group can be jointly recovered.}
Though \cite{14}  selected a group of relays to minimize the total consumed energy for a source-destination pair, it solved the relay selection by determining the optimum SNR threshold in an exhaustive search way.  {The selection of the optimum relay group is still an open problem for networks that are coded over Galois Field, especially when CSI is not available at transmitters and  inter-channels do not necessarily follow identical distribution (non-i.d.)}. 

Furthermore,  as we will show in Section V-D with more relays being selected, the diversity order can be increased and thus the energy consumption for the data transmission  can be lowered; however,  the electrical circuit energy consumption increases. This may decrease the energy efficiency. {In other words, a tradeoff  exists between the  diversity order and energy efficiency}. Thus  the circuit energy consumption is not negligible and should be taken into account \cite{3}. One effective way to reduce the energy costs is to adopt the sleeping/off-based mechanism at relays and BSs \cite{energysaving}, \cite{energysaving2}. That is, the selected relays and BSs are scheduled to be in receiving/transmitting/sleeping modes when necessary, while the unselected relays stay at the off mode during the whole transmission.



\subsection{Our Contributions}

We study the EE of MDNC based multi-user and multi-relay (MUMR) networks without CSIT.
We aim at maximizing the EE  meanwhile satisfying the outage probability constraint. Specifically, our main contributions are listed as follows:

\begin{description}

\item[(1)] {We explicitly find the connection among the outage probability, energy consumption and EE for the networks that are coded over Galois Field and formulate the maximizing EE problem.
    A joint power allocation (PA) and relay scheduling (RS) is conducted to solve the EE-maximization problem with satisfying the outage probability constraint. We note that the channel model is practical and general in the sense that the inter-channels are non-i.d. and the path-loss related to the transmission distance is also incorporated. All these lead to a completely new mathematical problem from the current state of arts}.

\item[(2)] {In our problem, the exact expression of the network outage probability is shown to be not tractable mathematically. To solve the optimization problem efficiently, outage probability is approximated by its log-convex form. Numerical results in Section V indicate that our approximation is sufficiently tight}.

\item[(3)] {EE is transformed into its subtractive-form based on the fractional programming theory. We transform our original problem into a convex mixed-integer nonlinear one, for which a generalized outer approximation (GOA) algorithm \cite{RFSL} is applied to decompose RS and PA such that they can be solved in an iterative manner. To our best knowledge, this is the first time when the GOA algorithm is used in similar network setups to efficiently determine an optimal relay group.}

\item[(4)] We provide an EE comparison between MDNC and NoNC. {It is shown that with the same number of selected relays,
MDNC can lead to EE gains.  However, if RS is implemented, NoNC can outperform
MDNC in terms of the EE when more relays are needed in the MDNC scheme.} Additionally, the impacts of relay number and  locations are also investigated. All these can give valuable insights on the system design of future wireless networks.

\end{description}

The rest of the paper is organized as follows. In Section II, we present the system model. Problem formulation is given in Section III. In Section IV, an approximated-optimal RS and PA algorithm is proposed. The convergence and complexity analyses are also given. The analytical and simulation results are presented in Section V. Section VI concludes this paper.

\section{SYSTEM MODEL}\label{sectionmodel}

We consider a common scenario in wireless networks, where $M$  users intend to transmit their independent messages {(i.e., a sequence of 0-1 bits)} to a BS with the assistance of $N$ micro relays. We make the following assumptions: (i) Every user has one message to be transmitted and all messages are of the same length\footnote{We note that this assumption is made for simplifying illustration. The system model can be  extended to general cases where different users may have different message lengths. {More specifically, if different users have different message lengths, we can divide the messages into shorter ones such that the lengths of shorter messages are the same and some users have more messages while some have fewer messages. Then, the users with fewer messages may not participate  in  all transmission rounds.}}; (ii) Without loss of generality, we assume that all the users and relays transmit information { with a fixed rate $\alpha_0$ bits per second (bps)} on every channel\footnote{This assumption is also made for simplifying illustration. Our model and algorithm are also applicable for different fixed rates on different channels. {The rates affect the values of the data transmitting time and  outage probability. However, different rates have no impact in the convexity of $\mathbf{P2}$ presented in Section IV. Hence, the analysis and proposed scheme are still feasible.} }; (iii) There is no direct connection between users and the BS due to the long distance or the presence of physical obstacles.

\subsection{{Transmission Scheme}}

We assume that all nodes operate in non-overlapping time slots adopted in \cite{17}, \cite{14} and \cite{25}. Let $U_i$, $i\in \{1, 2, \cdots, M\}$, and $R_j$, $j\in \{1, 2, \cdots, N\}$, respectively represent the $i$th user and $j$th relay. The whole transmission consists of two hops.
\begin{description}
  \item[\textit{1)\;The First Hop: User-relay Transmission}]
\end{description}
\begin{figure}[h]
\centering
\includegraphics[width=0.38\textwidth]{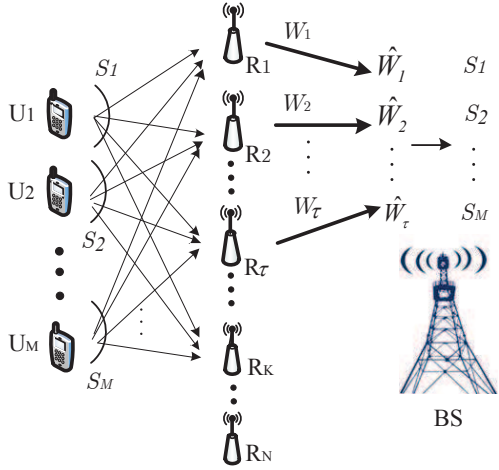}
\caption{$M$ users - $N$ relays - one BS network with MDNC.}\label{fig1}
\end{figure}

{The source message ${S_{i}}$ ($\forall i$) is first protected by channel coding and converted into a unit power-channel codeword, denoted as   
 $X({S_i})$ ($\forall i$).  In $U_i$'s broadcasting phase, $R_j$ receives the channel codeword $X({S_i})$ as follows}:
\begin{equation}
  Y_{i,j}= {h _{ij}}\sqrt {{p_i}} X({S_i}) + z_{ij}, \label{eq1}
\end{equation}
where ${h_{ij}}\sim \mathcal{CN}(0, {d_{ij}^{ - n_{ij}}\sigma _{{h_{ij}}}^2})$ is the channel gain, which combines path-loss and Rayleigh fading.
$d_{ij}^{ - n_{ij}}$ and $\sigma _{{h_{ij}}}^2$ denote the path loss and variance of the Rayleigh distribution, respectively; ${d_{ij}}$ is the distance; $n_{ij}$ is the channel path loss exponent;
$p_i$ is the transmitting power of $X({S_i})$; $z_{ij} \sim {\mathcal{N}}(0,N _{0,ij}B)$ denotes the additive Gaussian white noise; $N _{0,ij}$ is the one-sided power spectral density of white Gaussian noise and $B$ is the bandwidth.

The achievable rate for the channel between $U_i$ and $R_j$ is given by
\begin{align}
C_{ij}= B{\log _2}(1 + \frac{|{h_{ij}}{|^2}p_i}{{N_{0,ij}}B}), \label{eq0}
\end{align}
where $|{h_{ij}}|$ is the amplitude of ${h_{ij}}$.

Since all the channel codewords are of the same length (with the same length, i.e., $\alpha_0T$ bits) and transmitted with the same rate, the transmission of every channel codeword in the first hop will take  one time slot with the length of ${T}$ seconds. Correspondingly, every relay will receive all the $M$ channel codewords in $M{T}$ seconds.

After user transmission, the selected relay, e.g., $R_{j}$, will  decode $Y_{i,j}$, $\forall i \in \{1,2,\cdots,M\}$ and recover $S_{i}$. An outage event occurs in the channel between $U_i$ and $R_j$ when {the fixed data transmission rate, denoted as $\alpha_0$, is larger than the achievable rate $C_{ij}$.}

\begin{description}
  \item[\textit{2)\;The Second Hop:  Relay-BS Transmission}]
\end{description}

The following notations will be used in our description of the transmission process.

$u_j$:  It is defined as
\begin{equation}
{u_{j}}=\left\{\begin{array}{ll}
1, & \mathrm{R_j \; is \; selected},\nonumber\\
0, & \mathrm{otherwise.}\nonumber
\end{array} \right.
\end{equation}
We have ${\sum\limits_{j = 1}^N {{u_j}} }=||{\mathcal U}||_1$, where ${{{\mathcal{U}}}} = \{{{u_i}}, \forall i\}$ and $||\cdot|{|_1}$  is $1$-norm operation, which equals the number of ``$1$" in $\mathcal{U}$ and correspondingly represents the number of selec  ted relays.

${\Theta}$:  The index set of the selected relays with cardinality $||\mathcal U||_1$.

${{{\Phi _{K}}}}$:  Suppose $K$ relays succeed in {receiving $M$ channel codewords} (i.e., $X({S_1})$,  $X({S_2})$, $\cdots$,  $X({S_M})$) and  recovering all the users' messages (i.e., ${S_1}$,  ${S_2}$, $\cdots$, ${S_M}$)  by channel decoding in the first hop. The index set of these $K$ relays is collected into ${{{\Phi _{K}}}}$.

$\psi_{\tau}$: The index set of $\tau$ relays, which forward messages to the BS in the second hop.  ${\tau}=0$ means that no relay forwards messages to the BS.

Clearly, $\psi_{\tau} \subseteq {{{\Phi _{K}}}} \subseteq {\Theta}$.

As shown in Fig.~\ref{fig1}, if $R_j$ fails to decode any of the channel codewords, it will not forward signals but remains in the transmitting mode and  ready for forwarding the next decoded message\footnote{We note that if $R_j$, $\forall j \in {{{\Phi _{K}}}}$, decodes a part of the source messages before forwarding them to the BS, the user messages may be recovered at the BS. Under our relaying protocol, we obtain the upper bound of the outage probability and the corresponding lower bound of the EE.} 
. Otherwise, if it can decode all $M$ channel codewords (i.e., $X({S_1})$,  $X({S_2})$, $\cdots$,  $X({S_M})$) and  recover all the users' messages (i.e., ${S_1}$,  ${S_2}$, $\cdots$, ${S_M}$)  by channel decoding, MDNC  will be applied. A network codeword ${W_{j}}$ ($\forall j \in {{{\Phi _{K}}}}$) is generated at $R_{j}$ by the linear combination of $M$ source messages over a finite field GF($Q$), where $Q$ is the alphabet size. That is,
\begin{equation}
{W_{j}} = \sum\limits_{i = 1}^M{e_{i,j}}{S_i}, \forall j \in {{{\Phi _{K}}}},
 \label{wj} \end{equation}
where $e_{i,j}$ is the global encoding kernel for $S_{i}$ at relay $R_{j}$. $e_{i,j}$  converts $\log_2(Q)$  bits of one source message into one symbol. And it constitutes the transfer matrix ${\mathbf{H}_{||{\cal U}||_1 \times M}} $ corresponding to MDNC
\begin{equation}
{\mathbf{H}_{||{\cal U}||_1 \times M}} = \left( {\begin{array}{*{20}{c}}
{{e_{1,1}}}&{{e_{2,1}}}& \ldots &{{e_{M,1}}}\\
{{e_{1,2}}}&{{e_{2,2}}}& \ldots &{{e_{M,2}}}\\
{}& \ldots & \ldots &{}\\
{{e_{1, ||{\cal U}||_1}}}&{{e_{2,||{\cal U}||_1}}}& \ldots &{{e_{M,||{\cal U}||_1}}}
\end{array}} \right).
\nonumber \end{equation}
${\mathbf{H}_{||{\cal U}||_1 \times M}}$ has a rank $M$ for any $M$ columns \cite{19}. 

{${W_{j}}$ ($\forall j$) is a symbol sequence. One symbol can take any value in $\{ 0, 1, 2, ..., Q-1\}$.  Then $R_j$ regards ${W_{j}}$ as a sequence of information symbols and produces a channel codeword, $X({W_{j}})$,  which will be further forwarded to the BS}. 

{We note $X({W_j})$ itself is a unit-power channel codeword. Let ${{{p'_{j}}}}$ denote the transmitting power of $X({W_j})$}.  
 {With network coding, $X({W_{j}})$ and $X({S_{i}})$ ($\forall i, j$) have the same number of bits, i.e., $\alpha_0T$ bits (interested readers can refer to \cite{19}, \cite{20} for detailed illustration)}. Correspondingly, every transmission in the second hop uses $T$ seconds. 
 
Suppose $\tau$ relays manage to forward their codewords to the BS, as shown in Fig.~\ref{fig1}. The received channel codeword at the BS can be written as
\begin{align}
  {\hat W_j}={g _{j}}\sqrt {{{p'_{j}}}} X({W_j}) + z_{j}, \forall j \in {\psi_{\tau}}, \label{eq5}
\end{align}
where ${g_{j}}\sim \mathcal{CN}(0, {d_{j}^{ - n_{j}}\sigma _{{g_{j}}}^2})$ is the channel gain;
$d_{j}^{ - n_{j}}$ and $\sigma _{{g_{j}}}^2$ denote the path loss and variance of the Rayleigh distribution, respectively; ${d_{j}}$ is the distance; $n_{j}$ is the channel path loss exponent; $z_{j} \sim {\mathcal{N}}(0,{N_{0,j}}B)$ is the noise term and ${N_{0,j}}$ is the power spectral density of noise.

{After that, as depicted in Fig. \ref{fig1}, by conducting channel decoding at the BS, $\{{\hat W_j}, \forall j \in \psi _\tau \}$ are decoded into $\{W_j, \forall j \in \psi _\tau  \}$. Finally, the BS tries to obtain the source messages (i.e., ${S_1}$, ${S_2}$,  $\cdots$, ${S_M}$) jointly from $\{W_j, \forall j \in \psi _\tau\}$ by network decoding}.

\subsection{Power Consumption Model}

The power consumption at sources, relays and  BS in different modes are listed as follows:

\subsubsection{Power Consumption at Users}

 $p_i$ is upper bounded by the battery capacity $P_{S,max}$, i.e., $p_i\leq P_{S,max}$, $\forall i \in \{1,2,\cdots,M\}$. We denote the set of ${{p}_i}, \forall i$, as ${{\mathcal{P}}_{{S}}}$.

\subsubsection{Power Consumption at Relays}

The power consumptions  at the relays in sleeping, receiving and transmitting modes is described as follows \cite{3}:

\begin{small}
\begin{equation}
{P_{R}}=\left\{\begin{array}{ll}
{P_{0,R}} , & \mathrm{receiving\,\,mode,}\\
{P_{0,R}} + {\vartriangle _P}{{p'_{j}}}, & \mathrm{transmitting\,\,mode,}\\
{P_{sleep,R}}, & \mathrm{sleeping\,\,mode,} \label{powerR}
\end{array} \right.
\end{equation}
\end{small}
where ${\vartriangle _P}$ is the slope of the load-dependent power consumption. Generally, we have ${P_{0,R}}>{P_{sleep,R}}.$  {We note that the power used by the network coding is also included in ${P_{0,R}}$}. We assume ${{p'_{j}}}$ is upper bounded by $P_{R,max}$, i.e., ${{p'_{j}}} \leq P_{R,max}$. The set of ${{p'_{j}}}, \forall j$, is denoted as ${{\mathcal{P}}_{{MDNC}}}$.

\subsubsection{Power Consumption at the BS}

Different from the relays, the BS will never turn off since the restarting, reconfiguring and reloading are usually time consuming procedures. The power of the BS in receiving and sleeping  modes is given as
\begin{small}
\begin{equation}
{P_{BS}}=\left\{\begin{array}{ll}
{P_{sleep,{BS}}}, & \mathrm{sleeping\,\,mode\,\,in\,\,the\,\,first\,\,hop,}  \\
{P_{0,BS}} . & \mathrm{receiving\,\,mode\,\,in\,\,the\,\,second\,\,hop.}  \label{powerBS}
\end{array} \right.
\end{equation}
\end{small}
In general, we have ${P_{0,BS}}>{P_{sleep,BS}}.$ {We note that the power consumed for network decoding is also included in ${P_{0,BS}}$}.

\section{Problem Formulation}

\subsection{Problem Formulation}
For the network described in Section \ref{sectionmodel}, we aim at maximizing the EE by optimizing the data transmission power and scheduling the relay modes under the outage probability constraint.
{The formulation of the optimization problem can be formulated as}
\begin{align}{}
&\mathbf{P1:}\;\; [\mathcal{U}^*,\mathcal {P}_{S}^*,\mathcal {P}_{MDNC}^*] = \argmax_{[{{\mathcal{U }}},{{\mathcal{P}}_{S}},{{\mathcal{P}}_{{{MDNC}}}}]} {\eta _{EE}}  \label{eqob2} \hfill \\
& s.t. \nonumber\\
&\quad \quad {u_{j}}\in \{0,1\}, \forall j \in \{1,2,\cdots, N\},\label{eq31} \\
&\quad \quad u_j=I_{(0 ,P_{R,max}]}({{p'_{j}}})\mathop = \limits^\Delta\left\{ \begin{gathered}
  1\quad {p'_{j}} \in ({0},{P_{R,max}}], \hfill \\
  0\quad \textrm{otherwise},\quad  \hfill \\
  \end{gathered} \right. \label{eq29}\\
&\quad \quad {\Pr}_{out} \leq {{\Pr}_{out,target}},\label{eq41}\\
&\quad\quad {E_{BS,1}}+{E_{BS,2}}+{E_{R,1}}+{E_{R,2}}\leq{E_0}, \label{eq27}\hfill\\
&\quad\quad 0<{p_{i}}\leq{P_{S,\max}}, \forall i \in  \{1,2,\cdots, M\},\label{eq282} \hfill\\
&\quad\quad  0\leq{p'_{j}}\leq{P_{R,\max}}, \forall j \in  \{1,2,\cdots, N\},\label{eq28}
\end{align}
where $(\cdot)^\ast$ represents the optimum solution. $I(\cdot)$ is an indicator function. Inequation \eqref{eq29} describes the relationship between $u_j$ and ${{p'_{j}}}$, which indicates that if a relay is not selected, no power is allocated for data transmission.
Inequation \eqref{eq41} requires that the outage probability, denoted as ${{\Pr}_{out}}$ should be lower than the target level ${{\Pr}_{out,target}}$. Inequation  \eqref{eq27} means that the total energy at the BS (suppose ${E_{BS,1}}$ and ${E_{BS,2}}$ joules energy are respectively consumed in the first and second hop at the BS) and relays (suppose ${E_{R,1}}$ and ${E_{R,2}}$  joules energy are consumed in the first and second hop at the relays, respectively) should be less than the total available energy, $E_0$, which varies with the power supply condition in the power grid. It limits the maximum number of selected relays that the network can support.

In what follows, we give the detailed formulation for ${\eta _{EE}}$, ${{\Pr}_{out}}$, ${E_{BS,1}}$, ${E_{BS,2}}$, ${E_{R,1}}$ and ${E_{R,2}}$, respectively.

\subsection{Energy Efficiency}
The EE is evaluated as the expected number of successfully transmitted information bits $L$ divided by the total consumed energy ${E_{tot}}$ in the first and second hop  {\cite{5}, \cite{17}}, i.e.,
\begin{align}
{\eta _{EE}} &= \frac{L}{{{E_{tot}}}}.\label{EEMDNC0}
\end{align}
Since the network decoder either recovers all source messages or does not recover any message, the outage probability for all users in the MDNC scenario is the same. Thus, $L$ can be expressed as   
\begin{align}
L&=M{{\alpha_0}{T}(1-{{\Pr}_{out}})},
\end{align}
where ${{\Pr}_{out}}$ is the outage probability ({ similar definitions of the expected number of bits/throughput can be found in  \cite{bobai}, \cite{EVB}, \cite{MKSS})}.
Then \eqref{EEMDNC0} can be rewritten as
\begin{align}
\eta _{EE} &=\frac {M{{\alpha_0}{T}(1-{{\Pr}_{out}})}}{{{E_{tot}}}}.\label{EEMDNC}
\end{align}

\subsection{Total Consumed Energy}
 {We note that $E_{tot}$ takes into account all the energy consumption factors, including the states of relays/BS and allocated transmitting power. We first give its expression in \eqref{eq15}. The detailed illustrations are given as follows.}

In the first hop, $U_i$ consumes ${p_i}{T}$ joules to broadcast its message to the $||{\cal U}||_1$ selected relays. Every selected relay enters the receiving mode and consumes ${P_{0,R}}M{T}$ joules  in $M{T}$ seconds, while the remaining relays shut down to save energy. All the users and relays respectively consume $E_S$ and $E_{R,1}$ joules in the first hop. At the same time, the BS stays at the sleeping mode for $M{T}$ seconds and consumes $E_{BS,1}$ joules. $E_S$, $E_{R,1}$ and $E_{BS,1}$ are shown in \eqref{eq15}.
\begin{figure}[h]
\centering
\includegraphics[width=0.4\textwidth]{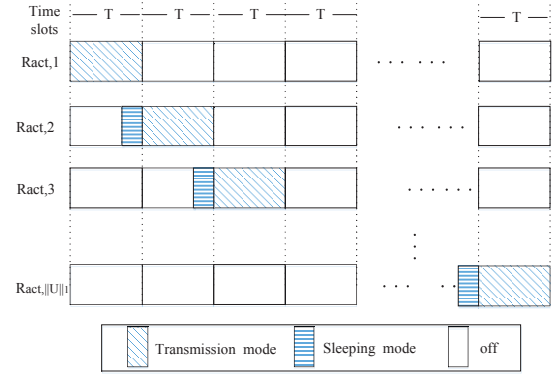}
\caption{A timing diagram for relaying transmission process and mode transition in the second hop for the MDNC scheme.}
\centering\label{transmission_NoNC}
\end{figure}

 In the second hop, the states of all relays are rescheduled. As shown in Fig. \ref{transmission_NoNC}, the selected relays are denoted as $R_{act,1}$, $R_{act,2}$, ...,$R_{act,||{\cal U}||_1}$, which  will forward their respective  codewords in a round-robin fashion. The selected relay, e.g., $R_j$  will consume $({P_{0,R}} + {\vartriangle_P}{{p'_{j}}}){T}$ joules energy to forward its codeword. Note that relays cannot enter the transmitting mode immediately from the off mode. Thus, to make the relay ready for immediately switching into the transmitting mode, it is necessary to first stay at a sleeping mode,  which acts as a transitional stage as shown in Fig. \ref{transmission_NoNC}. We assume the sleeping period lasts for $\beta T$ seconds, where $\beta$ is a constant and limited between $0$ and $1$. Specifically, during the last $\beta T$ seconds of one selected relay's transmission period, the next selected relay in the round-robin queue will be in the sleeping mode until its turn to be active and forward its codeword. Meanwhile, the other relays are off to reduce energy consumption and avoid interference. Thus, except the first selected relay (i.e., $R_{act,1}$), any other selected relay  will be in the sleeping mode for $\beta T$ seconds and consumes $\beta TP_{sleep,R}$ joules. The BS will be in the receiving mode in $||{\mathcal U}||_1$ time slots. In total, $E_{R,2}$ and ${{E_{BS,2}}}$ joules will be consumed respectively by relays and the BS as shown in \eqref{eq15}. 
 
\begin{figure*}[ht]
\begin{flalign}
{E_{tot}}&={E_S}+{E_R}+{E_{BS}}& \hfill\nonumber\\
&=\overbrace {\underbrace {\sum\limits_{i = 1}^M {{p_i}{T}} }_{{E_S}} + \underbrace {\sum\limits_{j = 1}^N {{u_j}{P_{0,R}}M{T}} }_{{E_{R,1}}} + \underbrace{{P_{sleep,BS}}M{T}}_ {{{E_{BS,1}}}} }^{1st\;hop} + \overbrace {\underbrace {\sum\limits_{j = 1}^N  {{u_j}}({P_{0,R}} + {\vartriangle_P}{{p'_{j}}}){T}+(\sum\limits_{j = 1}^N {{u_j}}  - 1){P_{sleep,R}}\beta {T}}_{{E_{R,2}}} + \underbrace {{P_{0,BS}}\sum\limits_{j = 1}^N {{u_j}}{T}}_{{E_{BS,2}}}}^{2nd\;hop}.& \label{eq15}
\end{flalign}
\hrulefill
\end{figure*}
\subsection{Outage Probability}
Since MDNC is used, any $M$ out of $\tau$ network codewords can be used to recover the $M$ user messages at the BS. An outage event happens when fewer than $M$ network codewords are received by the BS, i.e., $\tau \leq M$ \cite{19}. Thus, it is required that in every transmission, the number of transmitters should be smaller than the number of relays. If the  number of users exceeds $N$, we will split the users into several groups and transmission proceeds group by group. For convenience, we assume $M\leq N$.

Suppose $K$ relays  succeed in receiving all the source messages. As illustrated above, an outage event happens in the following two cases in terms of $K$.

\begin{description}
  \item[1. {Case} {$\mathcal{{A}}$:\;{$K< M$}}]
\end{description}

In this case, user messages cannot be recovered no matter how the second hop proceeds. Let $\Pr{\{ \zeta _{K} \}}$ denote the probability that ${K}$ relays successfully receive all the user messages in the first hop. An outage event happens when $K <M$ with probability ${\Pr_{{out},\mathcal{A}}},$ i.e.,
\begin{flalign}
&{{\Pr}_{{out},\mathcal{A}}}=\sum\limits_{K = 0}^{M - 1} {\Pr{{\{ \zeta_{K} \} }}}. \label{eqspouta}
\end{flalign}
Let ${\rho _j}$ represent the probability that ${R_j}$ can successfully recover all $M$ user messages in the first hop. It can be represented as
\begin{align}
& {\rho _j} = \prod\limits_{i = 1}^M {(1 - {{\Pr}_{e,ij}})}, \label{eq27ss}
\end{align}
where ${{\Pr}_{e,{ij}}}$ is the outage probability of $U_i$-$R_j$ end-to-end transmission and can be calculated as
\begin{align}
{{\Pr}_{e,{ij}}}&={\Pr}\{ C_{ij}  < {\alpha_0}\}={\Pr}\{ |{h_{ij}|^2} < \frac {({2^{{{\alpha_0} \mathord{\left/
 {\vphantom {R B}} \right.
 \kern-\nulldelimiterspace} B}}} - 1){N_{0,ij}}B}{ p_i}\}. \label{eq2}
\end{align}
Since $x_{ij}=|h_{ij}|^2$ follows the exponential probability density function (PDF) \cite{14}
\begin{eqnarray}
f(x_{ij}) = \frac{1}{{\sigma _{h_{ij}}^2d_{ij}^{ - {n_{ij}}}}}\exp ( - \frac{x_{ij}}{{\sigma _{h_{ij}}^2d_{ij}^{ -{n_{ij}}}}}),\nonumber
\end{eqnarray}
\eqref{eq2} can be rewritten as
\begin{align}
  {{\Pr}_{e,ij}} &= 1 - \int_\frac {({2^{{{\alpha_0} \mathord{\left/
 {\vphantom {R B}} \right.
 \kern-\nulldelimiterspace} B}}} - 1){N_{0,ij}}B}{ p_i}^{+\infty} f(x_{ij}){d{x_{ij}}}=1 - \exp ( - \frac{{{c_{ij}}}}{{{p_i}}}), \label{eq2CSI}
\end{align}
where  $c_{ij}=\frac{{({2^{{{{\alpha_0}} \mathord{\left/
 {\vphantom {{{\alpha_0}} B}} \right.
 \kern-\nulldelimiterspace} B}}} - 1) {{N_{0,ij}}B}}}{{d_{ij}^{ - n_{ij}}{\sigma _{h_{{ij}}}^2}}}>0$. Thus, $\Pr{\{ \zeta _{K} \}}$  can be represented as
\begin{align}
\Pr \{ {\zeta _K}\}  = \sum\nolimits_{{\Phi _K}} {\left( {\prod\limits_{{j} \in {\Phi _K}} {{\rho _j}} {\prod _{{j} \in {\Theta}\backslash {\Phi _K}}}(1 - {\rho _j})} \right)}, \label{eq37}
\end{align}
where {$K$ is the number of relays succeeding in receiving and decoding all the user messages in the first hop}; $\sum\nolimits_{{\Phi _K}} {(\theta)}$ represents the sum of $\theta$ when ${{\Phi _K}}$ is in different cases. ${{\Phi _K}}$ consists of $K$ relays randomly chosen from the selected  $||{\cal U}||_1$ relays and have $C_{||{\cal U}||_1}^K$ cases.

\begin{description}
  \item[2. Case $\mathcal{B}:$ \; $K \geq M$ ]
\end{description}

In this case,  an outage event  happens when the number of relays forwarding the codewords to the BS in the second hop is smaller than $M$.

Let ${{\Pr}_{{out},\mathcal{B}}}$ denote the outage probability in case $\mathcal{B}$, given  by
\begin{align}
&{{\Pr}_{{out},\mathcal{B}}}=\sum\limits_{K = M}^{||{\cal U}||_1} \left({\Pr{{\{ \zeta_{K} \} }} \cdot \sum\limits_{\tau  = 0}^{M - 1} \Pr{\{\varsigma_{\tau} |\zeta_K \}}}\right), \label{eqspoutb}
\end{align}
where $\Pr{{\{ \zeta_{K} \} }}$ is given in \eqref{eq37};  $\Pr{\{\varsigma_{\tau} |\zeta_K \}}$ represents the probability that any combination of $\tau$ relays in ${{{\Phi _{K}}}}$ successfully transmits the user message in the second hop, which is given as
\begin{align}
\Pr \{ {\varsigma _\tau }|{\zeta _K}\}  = \sum\nolimits_{{\psi _\tau }} {\left( {\prod\limits_{{j} \in {\psi _\tau }} {(1 - {{\Pr }_{e,j}})} {\prod _{{j} \in {\Phi _K}\backslash {\psi _\tau }}}{{\Pr }_{e,j}}} \right)}, \end{align}
where ``$A \backslash B$" means $A$ is the complementary set of $B$. $\sum\nolimits_{{\psi _\tau }} {(\varpi )}$  represents the sum of $\varpi$ when ${{\psi _\tau }}$ is in different cases. Since ${{\psi _\tau}}$  is the set of $\tau$ indexes randomly chosen from ${{\Phi _K}}$, including $C_K^{\tau}$ cases. ${{\Pr}_{e,j}}$   is  the outage probability for the transmission between $R_j$ and the BS. Similar to the calculation of ${{\Pr}_{e,ij}}$,  ${{\Pr}_{e,j}}$  can be evaluated as
\begin{align}
  {{\Pr}_{e,j}} &= 1 - \exp ( - \frac{{{c_j}}}{u_j p'_j}),\label{Pej}
\end{align}
where $c_j=\frac{{({2^{{\alpha_0}/B}} - 1){N_{0,j}}B}}{{{d_{j}^{ - n_{j}}\sigma _{{g_j}}^2}}} >0 $.

Since cases $\mathcal{{A}}$ and $\mathcal{{B}}$  are mutually independent, the outage probability can be calculated as
\begin{align}
&{{\Pr}_{out}}={{\Pr}_{{out},\mathcal{A}}}+{{\Pr}_{{out},\mathcal{B}}}. \label{eq36}
\end{align}

\section{Approximated-optimal Relay Scheduling and Power Allocation}

{As can be seen from \eqref{eqspouta} $-$ \eqref{eq36}, although ${{\Pr}_{out}}$ offers the exact expression of the network outage probability, it consists of multiple exponential items. Note the coefficients of exponential items are positive and negative constants that alternately appear. 

 Additionally, due to the existence of binary variables $u_{j}$, $\forall j \in \{1,2,\cdots,N\}$,  $\mathbf{P1}$ is a mixed integer nonlinear programming (MINLP) problem \cite{CAF} that is usually NP-hard \cite{NPhard}. {Furthermore, the product form of ${{u_j}{{{p'_{j}}}}}$ in \eqref{eq15} and \eqref{eq36}} and nonlinear fractional structure of ${\eta _{EE}}$ make $\mathbf{P1}$ more complicated and challenging. All these  makes \eqref{eq41} in $\mathbf{P1}$  not tractable mathematically. 
 
 To solve it efficiently, we will make following efforts to convert it into a convex  optimization problem such that the efficient algorithm can be adopted.  Specifically, three steps will be used including converting $\Pr_{out}$ and ${\eta _{EE}}$ into convex forms and  decomposing RS and PA into two-subproblems. More detailed steps are described as follows.}

\subsection* {Step 1: Geometric Programming Approximation of ${{\Pr}_{out}}$}

{Based on above observations, we are motivated to look into  the structure of the expression for outage probability in the high SNR region and find a simplified but  tight approximation. In the sequel, we manage to convert it into a geometric programming form, which is further converted into a log-convex form}. The detailed methods are illustrated as below.

We start with the following four items, i.e., ${\rho _j}$, $(1-{\Pr}_{e,j})$, $(1-{\rho _j})$ and ${\Pr}_{e,j}$ that appear in the expression of ${{\Pr}_{out}}$. By substituting \eqref{eq2CSI} into \eqref{eq27ss}, we have
\begin{align}
&{\rho _j}=  \exp ( - \sum\limits_{i = 1}^M {{c_{ij}}/{p_i}} ) ,\; 1-{\Pr}_{e,j}=\exp ( - {c_{j}}/({{u_j}{p'_{j}}})).
\end{align}
For high  SNR regions, i.e., when ${{c_{ij}}/{p_i}} \to 0, {c_{j}}/({{u_j}{p'_{j}}}) \to 0$, we have  ${\exp({ - {{c_{ij}}/{p_i}}}})$ $\sim 1$ and $\exp ( - {c_{j}}/({{u_j}{p'_{j}}}))\sim 1$. Furthermore, according to \eqref{eq27ss} and \eqref{Pej}, we have
\begin{align}
{\rho _j}\sim 1, (1-{\Pr}_{e,j}) \sim 1. \label{Poutapp}
\end{align}
Moreover, since $\mathop {\lim }\limits_{x \to 0} (1 - \exp ( - x)) = x$, we have
\begin{align}
& 1 - {\rho _j}=1-\exp ( - \sum\limits_{i = 1}^M {{c_{ij}}/{p_i}} )\approx\sum\limits_{i = 1}^M {\frac{{{c_{ij}}}}{{{p_i}}}}.\label{noout1}
\end{align}

${\Pr}_{e,j}=1-\exp ( - {c_{j}}/{{u_j}{p'_{j}}}) $ should be $1$ if ${p'_{j}}=0$. However, using  the approximation $\mathop {\lim }\limits_{x \to 0} (1 - \exp ( - x)) = x$ will lead to ${\Pr}_{e,j}$ infinitely large, and thus invalid. Instead, we take the approximation  $1 - {e^{ - x}}\sim \frac{x}{{x + 1}}$  and thus,
\begin{align}
&{\Pr}_{e,j}=\frac{{{c_j}}}{{{c_j} + {{u_j}{p'_{j}}}}}.  \label{noout2}
\end{align}
Note that  ${\Pr}_{e,j}=1$ when ${u_{j}}= 0$. By substituting \eqref{Poutapp} $-$ \eqref{noout2} into \eqref{eqspouta} $-$ \eqref{eq36}, we rewrite ${{\Pr}_{out}}$ as \eqref{eq56}, which is a tight approximation of ${{\Pr}_{out}}$  in the high SNR region.
\begin{figure*}[ht]
\begin{flalign}
&{{\Pr}_{out}}\approx \sum\limits_{K = 0}^{M - 1} {\sum\nolimits_{{\Phi _K}} {\left( {\prod\limits_{{j} \in {\Theta}\backslash {\Phi _K}} {\sum\limits_{i = 1}^M {\frac{{{c_{ij}}}}{{{p_i}}}} } } \right)} } + \sum\limits_{K = M}^{||{\cal U}||_1} {\left( {(\sum\nolimits_{{\Phi _K}} {(\prod\limits_{{j} \in {\Theta}\backslash {\Phi _K}} {(\sum\limits_{i = 1}^M {\frac{{{c_{ij}}}}{{{p_i}}}} } )} )) \cdot (\sum\limits_{\tau  = 0}^{M - 1} {\sum\nolimits_{{\Phi _K}} {(\prod\limits_{{j} \in {\Phi _K}\backslash {\psi _\tau }} {\frac{{{c_j}}}{{u_j{{{p'_{j}}}}}+c_j}} } ))} } \right)}.  & \label{eq56}
\end{flalign}
\hrulefill
\end{figure*}

The multiplicative form of integer variables (i.e., ${u_j}$, $\forall j \in \{1,2,\cdots,N\}$) and real variables ${{{p'_{j}}}}$ makes $\mathbf{P1}$ non-convex/concave. To covert into convex with variable substitution, two new variables, i.e., ${{\tilde p}_i}$ and ${{\tilde p'}_j}$ are first defined as below
\begin{align}
& {p_i} = {e^{{{\tilde p}_i}}},\quad \; {{u_j}{{p'_{j}}}} = {c_j}{e^{{\tilde p'}_j}}-{c_j},\label{eq489}
\end{align}
where $({c_j}{e^{{\tilde p'}_j}}-{c_j})$ is the real transmitting power at $R_j$.

Then, \eqref{eq282}, \eqref{eq28} and \eqref{eq56} can be rewritten as \eqref{eq63}, \eqref{eq631a} and \eqref{eq222}, respectively.
\begin{align}
&0<p_i={e^{{{\tilde p}_i}}}\leq P_{S,max},\label{eq63} \\
& 0 \leq {{u_j}{p'_{j}}}={c_j}{e^{{{\tilde p'}_j}}}-{c_j}{ \leq {{u_j}{P_{R,\max }}}}.\label{eq631a}
\end{align}
\begin{figure*}[t]
\begin{flalign}
&{{\Pr}_{out}}\approx\sum\limits_{K = 0}^{M - 1} {\sum\nolimits_{{\Phi _K}} {\left( {\prod\limits_{{j} \in {\Theta }\backslash {\Phi _K}} {\sum\limits_{i = 1}^M {{c_{ij}}{e^{ - {{\tilde p}_i}}}} } } \right)} } + \sum\limits_{K = M}^{||{\cal U}||_1} {\left( {(\sum\nolimits_{{\Phi _K}} {(\prod\limits_{{j} \in {\Theta }\backslash {\Phi _K}} {(\sum\limits_{i = 1}^M {{{c_{ij}}e^{ - {{\tilde p}_i}}}} } )} )) \cdot (\sum\limits_{\tau  = 0}^{M - 1} {\sum\nolimits_{{\Phi _K}} {(\prod\limits_{{j} \in {\Phi _K}\backslash {\psi _\tau }} {{e^{ - {{\tilde p'}_j}}}} } ))} } \right)}.\label{eq222}
\end{flalign}
\hrulefill
\end{figure*}
Now, ${{\Pr}_{out}}$ is approximated to a sum form of exponential items multiplied by positive constants, which is log-convex \cite{24}.

\subsection* {Step 2: Fractional Programming for the Objective}

{Note that the objective, i.e., $\eta _{EE}$ is still in a nonlinear fractional form, which cannot be guaranteed be to convex even though we ignore the existence of the $0-1$ variables. Following Step $1$, we further convert $\eta _{EE}$ into its subtractive form by conducting fractional programming based on the fractional programming theory \cite{nonl}}, we have the following lemma for solving $\mathbf{P1}$.

\begin{lemma}  \label{Theorem1}
 Define $V=M{\alpha_0}(1-{\Pr}_{out}) - {q}{E_{tot}}$, ${{\mathcal{\tilde P}}_{{MDNC}}}=\{{{\tilde p'}_j}, \forall j\}$ and ${{\mathcal{\tilde P}}_{{S}}}=\{{\tilde {p}_i}, \forall i\}$. If  $\mathcal{U}$ is fixed, the resource allocation policy can achieve the maximum energy efficiency
\begin{align} q^*=\max {\eta_{EE}},\nonumber
\end{align}
if and only if
\begin{align}
V( q^ *) &=V( q^ *, {{\mathcal{\tilde P}}_{{S}}}^*, {{\mathcal{\tilde P^*}}_{{MDNC}}}, {{\mathcal{U}}}^* ) \nonumber\\
&= \max \{M{\alpha_0}(1-{\Pr}_{out}) - { q^ * }{E_{tot}}\} = 0, \label{non}
\end{align}
where $V( q^ *)$ is referred to as the subtractive-form of the primal objective function; $q^ *$ is the maximum energy efficiency; ${{\mathcal{\tilde P}}_{{S}}}^*$, ${{\mathcal{\tilde P^*}}_{{MDNC}}}$ and ${{\mathcal{U}}}^*$ are the optimal solutions.
\end{lemma}

\begin{proof}
The proof is provided in Appendix A. \quad \end{proof}
According to Lemma \ref{Theorem1}, $\mathbf{P1}$ can be reformulated as finding the optimum transmitted powers at users and relays to satisfy
\begin{align}\max \{V =(M{\alpha_0}(1-{\Pr}_{out}) - {q}{E_{tot}})\}=0. \label{maxizero}\end{align}

In Dinkelbach's method, $q$ is iteratively updated \cite{nonl}. In every iteration, it solves \eqref{maxizero} with a given $q$ and then judges whether $V(q)$  converges to a given tolerance. If not, $q$ is updated and we repeat the maximization problem until it converges or reaches the maximal iterations. The details can be found in \cite{nonl}, \cite{OFDM_nonl}.
Note that with given $q$ in every iteration, we have
 \begin{align}
 \max \;V \Leftrightarrow \min \; - V \Leftrightarrow \min \;\;V' \mathop \Leftrightarrow \limits^{(a)} \min \;\;\tilde V,\label{eqrelation}
\end{align}
where  ``$\Leftrightarrow$" means the expression in the left side is equivalent to that in the right side. $V'$ and $\tilde V$ are respectively defined as
\begin{align}
V'&=-V+M{\alpha_0}+q{T}{{\Delta _P}}\sum\limits_{j = 1}^N {{c_j}},\label{Vprime}\\
\tilde V& = \log(V')\label{tildeV}.\end{align}
$(a)$ satisfies with the fact that $V'>0$.

By taking log for both sides of \eqref{eq41}, we have $\mathbf{P2}$ shown below.
\begin{align}{}
&\mathbf{P2}:  \mathop {\min }\limits_{{{{\mathcal{\tilde P}}}_S},\;{{{\mathcal{\tilde P}}}_{{{MDNC}}}},\;{\mathcal{U}}} \;\tilde V \nonumber\\
& s.t.\; \eqref{eq27},\; \eqref{eq63}, \; \eqref{eq631a}, \nonumber\\
&\quad\quad \log({{\Pr}_{out}}) \leq \log({{\Pr}_{out,target}}). \label{poutfinal}
\end{align}
According to \eqref{maxizero}-\eqref{tildeV}, the optimal solutions, i.e., ${{\mathcal{\tilde P}}_{S}}^*$, ${{\mathcal{\tilde P^*}}_{{MDNC}}}$ and ${\mathcal{U^*}}$ must satisfy
\begin{align}\min \{\tilde V\} =\log(M{\alpha_0}+q{T}{{\Delta _P}}\sum\limits_{j = 1}^N {{c_j}}).\end{align}

Finally, we have the following {proposition} for $\mathbf{P2}$.
\begin{prop}  \label{prob1}
With  given $q$, $\mathbf{P2}$ is a convex optimization problem with respect to ${\tilde p}_i$, $\forall i \in \{1,2,\cdots, M\}$ and ${{\tilde p'}_j}$ $\forall j \in {{{\Phi _{K}}}}$. Moreover, the existence of the $0-1$ variable, i.e., $u_j$, $\forall j$ makes it  also a  convex MINLP.
\end{prop}

\begin{proof}
The proof can be found in Appendix $B$.
\end{proof}

\subsection * {Step 3: Iterative Algorithm for PA and RS}
 
{In Proposisition   
 \ref{prob1},  we state that our optimization problem has been converted into a convex  MINLP. Based on that, in this section, we seek the optimum solutions of $\mathbf{P2}$ with the GOA algorithm \cite{CAF}, which is proposed to efficiently solve the MINLP problem in \cite{CAF} and \cite{RFSL}. We use the optimum power allocation (PA) solutions obtained in the prior iterations to further optimize
relay selection (RS) in the sequent iterations. Specifically, we separate $\mathbf{P2}$ into two subproblems, a primal problem (PP) and a master problem (MP). PP focuses on obtaining ${{\mathcal{\tilde P}}_{S}}$ and ${{\mathcal{\tilde P}}_{{MDNC}}}$ with obtained $\mathcal{U}$ in the prior iteration, while the master problem aims at updating the $0-1$ combinations of $\mathcal{U}$ with previously obtained ${{\mathcal{\tilde P}}_{S}}$, ${{\mathcal{\tilde P}}_{{MDNC}}}$ and $\mathcal{U}$. The two optimization problems are solved iteratively.}

\subsubsection * {Primal problem in the $t$th iteration}
In the $t$th iteration, the primal problem of $\mathbf{P2}$ is formulated as
\begin{align}{}
&\mathbf{PP}:  \min\; \tilde V({q^{(\theta )}},{{{\mathcal{\tilde P}}}_{S}}^{},{{{\mathcal{\tilde P}}}^{}}_{{{MDNC}}},{{{\mathcal{U}}}^{(t - 1)}})  \nonumber\\
& s.t.\quad \eqref{eq27},\; \eqref{eq63}, \; \eqref{eq631a},\; \eqref{poutfinal}, \nonumber
\end{align}
where ${q^{(\theta )}}$ is the newly updated value of $q$ in its ${{\theta}}$th iteration; ${{{\mathcal{U}}}^{(t - 1)}}$ is the optimum 0-1 combinations of ${{{\mathcal{U}}}}$ in the $(t - 1)$th iteration.
{Proposition $1$ have revealed that $\mathbf{PP}$ is a standard convex optimization, which can be solved efficiently with the interior point method \cite{24}}.

Additionally, let $UBD_{}^{(t)}$ stand for the upper bound of ${\tilde V}$ provided by the primal problem within the first $t$ iterations. $UBD_{}^{(t)}$ is set as the minimum value of $\tilde V$ obtained in the current and the prior iterations, i.e.,
\begin{align}
UBD_{}^{(t)} =  \min\;\{{\tilde V}^{(m)}\}, \forall m \leq t.\label{UBD}
\end{align}

\subsubsection * {Master problem in the $t$th iteration}
{In the $t$th iteration, the master problem aims at updating the 0-1 combinations  in $\mathcal{U}^{(t)}$, which will be used in the $(t+1)$th primal problem iteration. Recall that the  optimal 0-1 combination corresponds to  the best RS strategy.}  We formulate it as
\begin{align}{}
&\mathbf{MP}:  \mathop {\min }\limits_{{{{\mathcal{\tilde P}}}_S},{{{\mathcal{\tilde P}}}_{MDNC}},{\mathcal{U}},w_{GOA}^{}} w_{GOA}  \nonumber\\
& s.t.\;\eqref{eq27},\; \eqref{eq63}, \; \eqref{eq631a}, \nonumber\\
&\quad\quad w_{GOA}^{} \ge \tilde V({q^{(\theta )}}, {\mathcal{\tilde P}}_S^{(m)},{{{\mathcal{\tilde P}}}^{(m)}}_{{{MDNC}}},{{{\mathcal{U}}}^{(m - 1)}}) \nonumber\\
&\quad\quad \quad\quad \quad\; + \nabla \tilde V{({\mathcal{\tilde P}}_S^{(m)},{{{\mathcal{\tilde P}}}^{(m)}}_{{{MDNC}}},{{{\mathcal{U}}}^{(m - 1)}})^T}\nonumber\\
&\quad\quad \quad\quad \quad\;\cdot \left( \begin{array}{l}
{{{\mathcal{\tilde P}}}_S} - {\mathcal{\tilde P}}_S^{(m)}\\
{{{\mathcal{\tilde P}}}_{MDNC}} - {\mathcal{\tilde P}}_{MDNC}^{(m)}\\
{\mathcal{U}} - {{{\mathcal{U}}}^{(m - 1)}}
\end{array} \right), \label{cons1}\\
&\quad\quad 0 \ge g(\tilde {\cal P}_S^{(m)} ,\tilde {\cal P}_{MDNC}^{(m)} )
 \nonumber\\
&\quad\quad \quad\quad \; + \nabla g(\tilde {\cal P}_S^{(m)},\tilde {\cal P}_{MDNC}^{(m)} )^T \nonumber\\
&\quad\quad \quad\quad\;\cdot \left( \begin{array}{l}
 \tilde {\cal P}_S  - \tilde {\cal P}_S^{(m)}  \\ 
 \tilde {\cal P}_{MDNC}  - \tilde {\cal P}_{MDNC}^{(m)}  \\ 
 \end{array} \right), \label{linearcons}\\
&\quad\quad UBD_{}^{(t)} > w_{GOA}^{}, \label{cons2}\\
&\quad\quad {||{\cal U}||_1^{(low)}}\leq ||{\mathcal{U}}|{|_1}\leq {||{\cal U}||_1^{(up)}}, \label{cons3}
\end{align}
where ${\mathcal{\tilde P}}_S^{(m)}$ and ${\mathcal{\tilde P}}_{MDNC}^{(m)}$ are the previously obtained optimal solutions of the primal problem in the $m$th ($m\leq t$) iteration; ${{{\mathcal{U}}}^{(m - 1)}}$ is the fixed 0-1 assignment that was obtained in the $(m-1)$th iteration. $\nabla \tilde V(\cdot)$ is the partial derivative of $\tilde V$. Take $\nabla \tilde V({\mathcal{\tilde P}}_S)$ as an example. It is the $M$-dimensional vector of the partial derivatives of $\tilde V$ with respect to ${\tilde p_1}$,  ${\tilde p_2}$, $\cdots$,  ${\tilde p_M}$. 
$g(\cdot)$ in \eqref{linearcons} describes the outage probability constraint, i.e., $g={{\Pr}_{out}}-{{\Pr}_{out,target}}$. \eqref{cons1} and \eqref{linearcons} respectively  linearise the objective  and the constraint around the obtained solution point (i.e., $\tilde {\cal P}_S^{(m)}$, $\tilde {\cal P}_{MDNC}^{(m)}$, $m=1, 2,\cdots, t$) of the primal problem, $\mathbf{PP}$. 

Additionally, we note that $\mathcal{U}^{(t)}$ obtained only with the constraints in \eqref{cons1}-\eqref{cons2} may make the next primal problem (i.e., $\mathbf{PP}(t+1)$) infeasible. That is, some constraints cannot be met \cite{CAF}. Thus, unnecessary computation is carried out if  the obtained $\mathcal{U}^{(t)}$ is utilized in the $(t+1)$th primal problem iteration. To avoid that, the feasible region for $||{\cal U}||_1$ is given as \eqref{cons3}. Specifically,  ${||{\cal U}||_1^{(low)}}$ stands for the minimum required number of relays to satisfy the outage probability constraint. ${||{\cal U}||_1^{(low)}}$ is determined by \eqref{Xout1} and \eqref{Xout2} as follows:
\begin{align}
{{\Pr} _{out,mp}}(||{\cal U}||_1^{(low)}) &\leq {{\Pr} _{out,target}}, \label{Xout1}\hfill \\
{{\Pr} _{out,mp}}(||{\cal U}||_1^{(low)} - 1) &> {{\Pr} _{out,target}}, \label{Xout2}
\end{align}
where ${\Pr _{out,mp}}(||{\cal U}||_1^{(low)})$ is the outage probability when all the users and $||{\cal U}||_1^{(low)}$ relays transmit data with their maximum possible power. Further, let
$||{\cal U}||_1^{(up)}$ represent the upper bound number of relays, which is defined as
\begin{align}||{\cal U}||_1^{(up)}=\min (N, ||{\cal U}||_1^{(PC)}), \label{Xup}
\end{align}
where $||{\cal U}||_1^{(PC)}$ is the maximum number of relays that can be supported by the total available energy in the constraint \eqref{eq27}.

It can be observed that,  appart from \eqref{cons1} and \eqref{linearcons}, the other constraints in $\mathbf{MP}$, including \eqref{eq27}, \eqref{eq63}, \eqref{eq631a}, \eqref{cons2}  and \eqref{cons2}  are also (or can be easily converted into) linear inequations. {Thus, it can be claimed that $\mathbf{MP}$ is a mixed-integer linear problem (MILP), for which the branch-and-cut algorithm in the reliable solver such as Cplex can be applied  \cite{CAF}.}

Note that, as illustrated in \cite{CAF}, $w_{GOA}^{(t)}$ obtained in the $t$th iteration can be considered as a lower bound of $\tilde V$. We denote it as $LBD^{(t)}$.

In Algorithm $1$, we summarize the algorithm of RS and PA strategies.

\begin{figure}[h]
\centering
\includegraphics[width=0.5\textwidth]{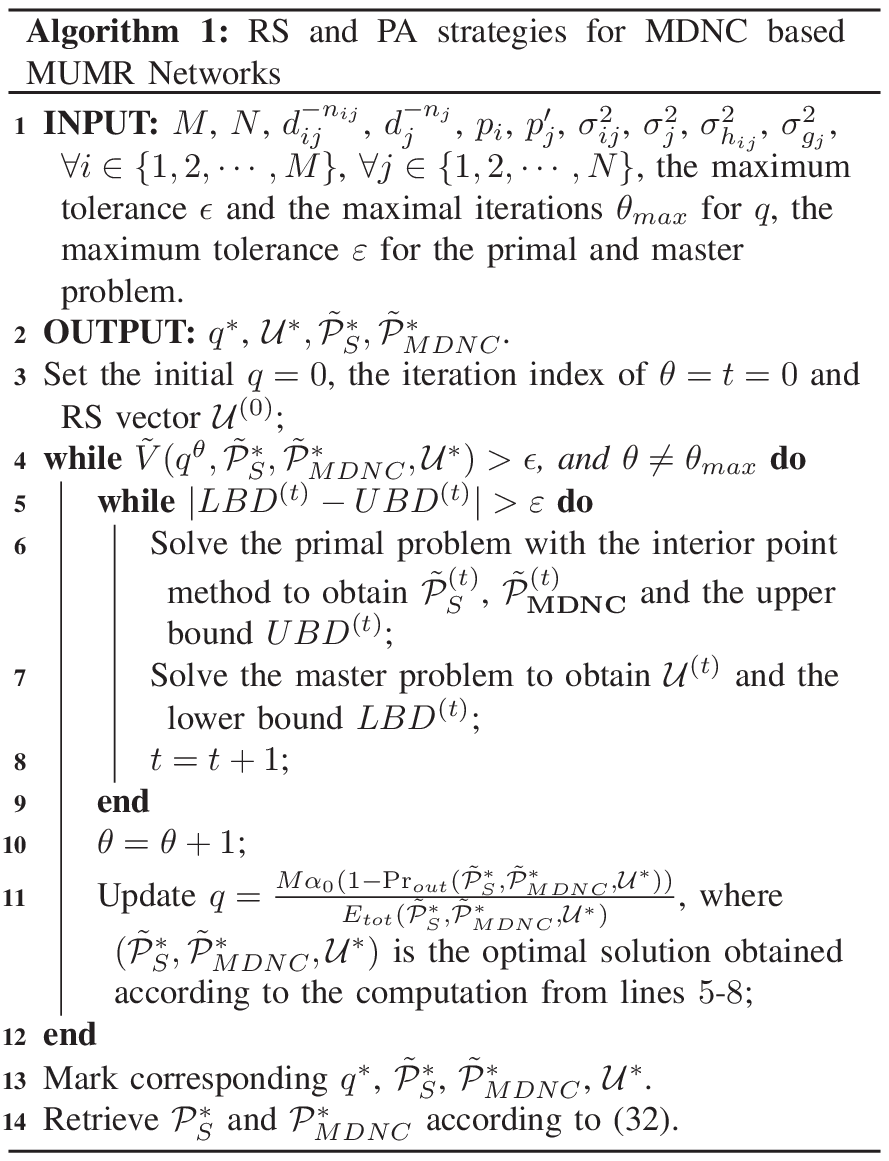}
\end{figure}

{\emph{Convergence Analysis:} We first analyze the monotonicity of the upper bound in the consecutive iterations. The primal subproblem (i.e., $\mathbf{PP}$) can be solved with the efficient algorithm with given ${{\mathcal{U}}}$. It provides an upper bound (i.e., $UBD$) of the original problem $\mathbf{P2}$. The upper bound is updated according to \eqref{UBD}, which guarantees that the sequence of the updated upper bound is non-increasing. For the lower bound (i.e., $LBD$), the constraint in \eqref{cons1} guarantees that $w_{GOA}^{(t)}\geq w_{GOA}^{(m)}$, $\forall m\leq t$. Thus, the lower bound sequence is non-decreasing. Considering the aforementioned analysis, it can be claimed that the GOA algorithm converges with finite iteration times.}

{To have an insight of the convergence rate, we further  give an analysis on the computational complexity in what follows}. 

{\emph{Complexity Analysis:} With our algorithm, PA and RS  were decomposed and respectively converted
into the standard convex $\mathbf{PP}$ and mixed-integer linear  
 $\mathbf{MP}$, which can  be
efficiently solved with interior point method and branch-and-cut
algorithm.} Specifically, in Dinkebach's method, the iteration time for $q$ is limited \cite{nonl}. Furthermore, with the interior-point method applied for $\mathbf{PP}$, the complexity will be $\mathcal{O}(\mathcal{C}_1\mathcal{C}_2)$, where $Z$ is the total number of exponential terms in the objective and constraints, $\mathcal{C}_1={(M + ||{\mathcal{U}}|{|_1} + Z + 1)^{1/2}}$, $\mathcal{C}_2=(M + ||{\mathcal{U}}||_1 + 1){{Z}^2} + {{Z}^3} + {(M + ||{\mathcal{U}}|{|_1})^3}$ \cite{interior_complexity}.  Obviously, $\mathcal{O}(\mathcal{C}_1\mathcal{C}_2)$ is a polynomial in $M$, $N$, and $||{\mathcal{U}}|{|_1}$. Due to the existence of the integer RS variables, the computational complexity of $\mathbf{MP}$ may be non-polynomial. However, efficient optimization solvers, e.g., Cplex, Mosek or Mskgpopt can be applied to solve the  master problem.

\section{Numerical Results}

In what follows, we will give numerical results based on above theoretical analyses and simulations. The following parameters are assumed in our system.
${\mathbf{{\sigma _h}}}$ and ${\mathbf{{\sigma _g}}}$ denote the variance {matrices} of the Rayleigh fading gains; ${\mathbf{{N_{0,h}}}}$ and  ${\mathbf{{N_{0,g}}}}$ represent the power spectrum density matrices respectively for S-R and R-BS channels, which are measured in Watts/Hz; ${\mathbf{{d _h}}}$ and ${\mathbf{{d _g}}}$ represent the distance matrices which are measured in meter; ${\mathbf{{n _h}}}$ and ${\mathbf{{n _g}}}$ represent the path-loss exponent {matrices}.  Specifically, elements at the $i$th row and the $j$th column of ${\mathbf{{\sigma _h}}}$, ${\mathbf{{d _h}}}$, ${\mathbf{{n _h}}}$ and ${\mathbf{{N_{0,h}}}}$ corresponding to the parameters for the channel between $U_i$ and  $R_j$. Elements at the $j$th columns in ${\mathbf{{\sigma _g}}}$, $\mathbf{{d _g}}$, ${\mathbf{{n _g}}}$ and ${\mathbf{{N_{0,g}}}}$ corresponding to the parameters for the channel between $R_j$ and the BS. The other system parameters  are shown in TABLE I. ${P_{0,BS}}$, ${P_{sleep,BS}}$, ${\Delta_P}$, $P_{sleep,R}$  and ${P_{0,R}}$ are also used in \cite{3}. The randomly generated values for all  above parameters are as follows.
\[\begin{gathered}
  {{\mathbf{\sigma }}_{\mathbf{h}}} = \left[ {\begin{array}{*{20}{c}}
  {5.1291}&{3.5040}&{4.3367}&{1.1597} \\
  {4.6048}&{0.9505}&{7.0924}&{0.7808}
\end{array}} \right], \hfill \\
  {{\mathbf{\sigma }}_{\mathbf{g}}} = \left[ {\begin{array}{*{20}{c}}
  { 5.0213}&{4.6821}&{3.4823}&{0.8667}
\end{array}} \right], \hfill \\
  {{\mathbf{d}}_{\mathbf{h}}} = \left[ {\begin{array}{*{20}{c}}
  {857.5}&{457.6}&{927.1}&{840.2} \\
  {1064.8}&{990.5}&{435.3}&{161.8}
\end{array}} \right],\hfill \\
 {{\mathbf{n}}_{\mathbf{h}}} = \left[ {\begin{array}{*{20}{c}}
  { 2.5570}&{2.9150 }&{2.3152}&{3.0143} \\
  {3.0938 }&{3.1298}&{2.6412}&{2.9708}
\end{array}} \right], \hfill \\
  {{\mathbf{d}}_{\mathbf{g}}} = \left[ {\begin{array}{*{20}{c}}
  {321.7}&{895.7}&{752.2}&{929.4}
\end{array}} \right],\hfill \\
{{\mathbf{n}}_{\mathbf{g}}} = \left[ {\begin{array}{*{20}{c}}
  {2.8006}&{2.2838}&{2.8435 }&{2.5315}
\end{array}} \right], \hfill \\
  {{\mathbf{N}}_{{\mathbf{0}},{\mathbf{h}}}} = 10^{-14}\left[ {\begin{array}{*{20}{c}}
  {0.063}&{0.035}&{0.27}&{0.003} \\
  {0.001}&{0.001}&{0.214}&{0.548}
\end{array}} \right],\hfill \\
  {{\mathbf{N}}_{{\mathbf{0}},{\mathbf{g}}}} = 10^{-14}\left[ {\begin{array}{*{20}{c}}
  {0.1900}&{0.3621}&{0.0132}&{0.0612}
\end{array}} \right].
\end{gathered} \]

\begin{table}[h]
\centering
 \caption{System Parameters}
\begin{tabular}{|c|c||c|c|}
\hline
$M$                                   & 2     & $N$    & 4\\
\hline
${P_{0,R}}$                                   & 56 watts     & ${P_{S,max}}$    & 10 watts  \\
\hline
$B$                  & 125K  Hz     & ${P_{R,max}}$    & 20 watts  \\
\hline
${P_{0,BS}} $                     & 130 watts       & ${P_{sleep,BS}}$ & 75 watts  \\
\hline
${\Delta_P}$                               & 2.6     & $P_{sleep,R}$     & 39 watts  \\
\hline
$\alpha_0$                               & 300K     & $|{X_{S,i}}|$     & 125K  \\
\hline
$\beta$                               & 0.1    & $E_0$      & $900$ Joules  \\
\hline
$||{\mathcal U}||_1^{(up)}$                               & 4    &      &  \\
\hline
\end{tabular}
 \label{tab11}
\end{table}

For comparison, we also draw the curves for the TDMA NoNC scenario.
The transmission process is similar to that in the MDNC scenario shown in Fig. \ref{transmission_NoNC}.  A sleeping period with length $\beta T$ seconds is also reserved for the relays. The difference from the NoNC scenario is that every selected relay will use $MT$ seconds instead of $T$ seconds to forward all the $M$ users' messages one by one. The power is fixed during the $MT$ seconds. An outage event happens for $U_i$, $\forall i$, when no relay successfully receives and forwards its message to the BS.

In the following, we first provide the relay selection results. After that, we illustrate the EE gains from RS, PA and  MDNC, respectively. In particular, the EE gains provided by dynamic switching between MDNC and NoNC when RS is implemented are also depicted. In the last, we show the impacts of circuit energy costs on RS and the impact of the relay locations on the EE.

\subsection{Relay Selection Results}
TABLE II shows the relay selection results for NoNC and MDNC scenarios. ``$\times$" marks the unselected relays.
We can see that though there are $4$ available relays, only parts of them are selected.  Note that with the same target outage probability, more selected relays are generally needed for the MDNC scenario than that in the NoNC scenario. 

\begin{figure}[h]
\centering
\includegraphics[width=0.4\textwidth]{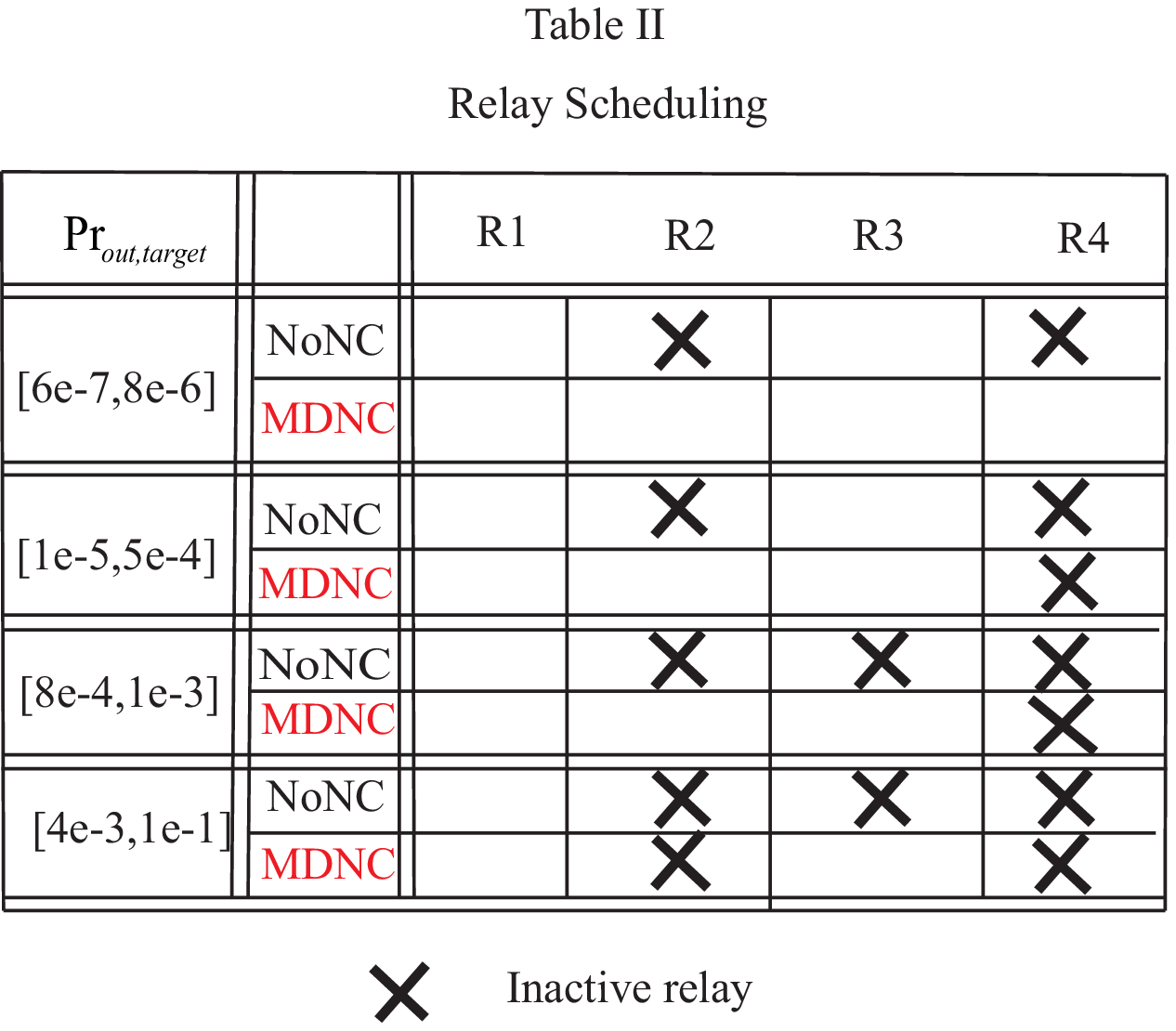}
\end{figure}
\subsection{Verification of approximated-optimal RS and PA}
Fig. \ref{fig7} demonstrates the Pareto-optimal tradeoff curves between  ${\eta _{EE}}$ and $\Pr_{out,target}$ for the MDNC and NoNC scenarios which are obtained by Algorithm $1$, the Brute-force algorithm and simulation, respectively. Pareto-optimal tradeoff curves consist of Pareto-optimal points, at which ${\eta _{EE}}$ cannot be increased without increasing $\Pr_{out,target}$ and vice versa. Moreover, the numerical results imply  that,  in the high SNR region, the Pareto-optimal ${\eta _{EE}}$ is obtained when $\Pr_{out}=\Pr_{out,target}$.  The simulation results are obtained by respectively averaging the outage probability, total consumed energy and EE over $10^8$ random realizations of the Rayleigh fading channels. It can be observed that the results from Brute-force, Algorithm $1$ and the simulations are closely matched. All these show that the analytical results obtained by Algorithm $1$ are valid.

\subsection{EE Gains from RS}
\subsubsection{EE Gains from RS for the MDNC scenario}

For comparison, we  plot ``MDNC+PA+RS" (MDNC with PA and RS), $4$-relay ``MDNC+PA" (MDNC with PA) and $3$-relay ``MDNC+PA" curves in Fig. \ref{fig7}. Due to the relay selection, jumping points appear in the curve for the scenario with RS. We can see from Table II that: at these jumping points, the number of selected relays varies from $2$ to $3$ and then $4$  when ${\Pr_{out}}$ decreases. Correspondingly, the electrical circuit consumption at the relays sharply increases. This results in the abrupt decrease in the EE.

Specifically, the curves of ``MDNC+PA+RS"  and $4$-relay ``MDNC+PA"  overlap when  $5\times10^{-6} \leq \Pr_{out} \leq 10^{-5}$. This is because in that region, both cases need $4$ relays  to forward messages. Meanwhile,  ``MDNC+PA+RS" outperforms $4$-relay ``MDNC+PA" when $10^{-5} \leq \Pr_{out} \leq  10^{-2}$. To be specific, more than $40\%$ of EE gain is achieved by RS.

For $3$-relay ``MDNC+PA", we note the lower bound of $\Pr_{out}$ is  $10^{-5}$. Correspondingly, $3$-relay ``MDNC+PA" curve ends when $\Pr_{out}=10^{-5}$. It overlaps with ``MDNC+PA+RS" when $10^{-5} \leq \Pr_{out} \leq 10^{-3}$ and both cases need $3$ relays.

{When $ \Pr_{out}\geq 4\times10^{-3}$, only two relays are selected in the ``MDNC+PA+RS". It is worth noticing that $2$-relay ``MDNC+PA" curve overlaps with that of $2$-relay ``MDNC+PA+RS" when $\Pr_{out} \geq 4\times 10^{-3}$. The outage probability beyond that region is not achievable with $2$-relay ``MDNC+PA+RS" or $2$-relay ``MDNC+PA" scheme.} Correspondingly, more than $28\%$ of EE gain is achieved by  ``MDNC+PA+RS"  compared with  $3$-relay ``MDNC+PA". {We note that, for $1$-relay scenario, network coding is not needed since the relay is the only path for the source to reach the destination. }

{We also draw the tradeoff curves when relays decoding a part of source messages are allowed to forward the corresponding codewords. As we can see, a higher EE can be obtained compared with our relay forwarding rule. But the extra EE gains are small.}


\begin{figure}
\centering
\includegraphics[width=0.55\textwidth]{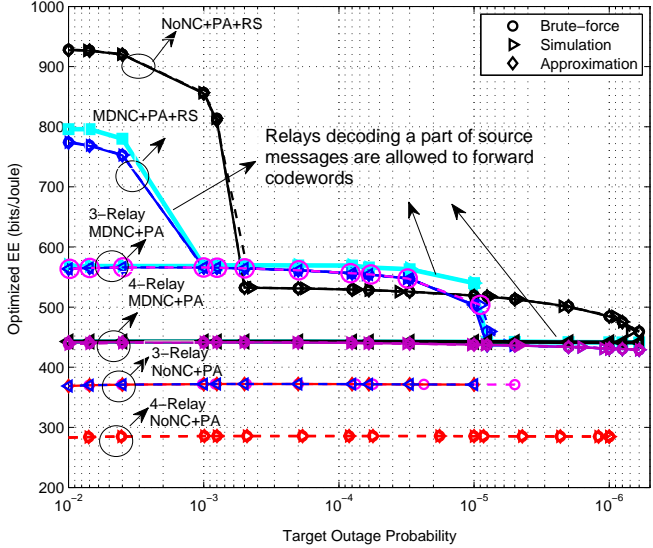}
\caption{Tradeoff between outage probability and EE for MDNC and NoNC scenarios.}  \label{fig7}
\end{figure}

\begin{figure}
\centering
\includegraphics[width=0.52\textwidth]{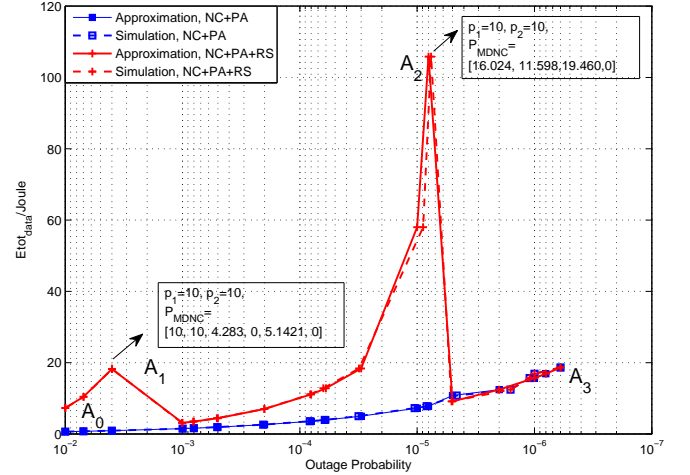}
\caption{$Etot_{data}$ versus ${\Pr}_{out}$ for MDNC scenario.}\label{fig8}
\end{figure}
To illustrate the impacts of $\Pr_{out}$ on the RS and PA, we plot the curves of total consumed energy used for the data transmission, denoted as $Etot_{data}$, versus $\Pr_{out}$ in Fig. \ref{fig8}. The curve for ``MDNC+PA+RS"  fluctuates and two peak points $A_1$ and $A_2$ appear. They correspond to the two jumping points in Fig. \ref{fig7}. We can see that at these two points, the power allocated to the users both reaches the upper limit, i.e. $10$ Watts. For the curve segment connecting $A_0$ and $A_1$, only $R_1$ and $R_3$ are selected. When $\Pr_{out}$ decreases, the power allocated to the users and relays becomes larger and reaches the upper bound at $A_1$. To further decrease $\Pr_{out}$, another relay, $R_2$, is selected to increase the diversity order. Consequently, $Etot_{data}$  decreases. When the power allocated to the users achieves the upper bound again at point $A_2$, the relay with the lowest priority, i.e., $R_4$, is selected to further increase the diversity order and decrease $\Pr_{out}$. Consequently, data transmitting power is diminished.
\subsubsection{EE Gains from RS for the NoNC scenario}

Due to the relay selection, jumping points appear in the curve for the NoNC scenario with RS. We can see from Table II that at these jumping points, the number of selected relays varies from $1$ to $2$ when $\Pr_{out,target}$ decreases.

In Fig. \ref{fig7}, by comparing ``NoNC+PA+RS" (NoNC scheme with PA and RS) with  $4$-Relay ``NoNC+PA" (NoNC scheme with PA), we can see that at the same outage probability level, more than $80\%$ of EE gains can be obtained from RS. The gap between the  ``NoNC+PA+RS"   and $3$-Relay ``NoNC+PA" implies that around $40\%$ of EE gains can be achieved by RS.

We also plot the curve of  $Etot_{data}$ versus $\Pr_{out}$ for ``NoNC+PA+RS" in Fig. \ref{fig11}. The curve for ``NoNC+PA+RS"  fluctuates and there is also one jumping point. It corresponds to the jump point in Fig. \ref{fig7}. For the curve connecting $B_0$ and $B_1$, only $R_1$ is selected. At $B_1$, the power allocated to the first user reaches its upper limit, i.e. $10$ Watts. To further decrease $\Pr_{out}$, another relay, $R_3$, is selected to increase the diversity order. Consequently, the total data transmitting powers can be diminished.

\begin{figure}
\centering
\includegraphics[width=0.56\textwidth]{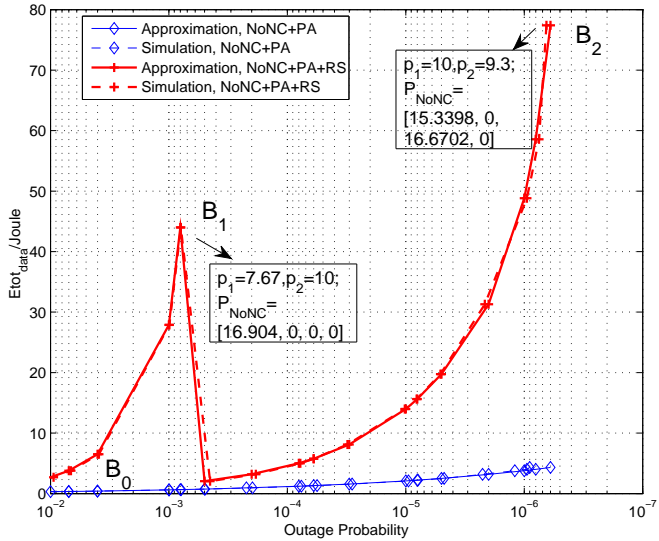}
\caption{$Etot_{data}$ versus ${\Pr}_{out}$ for NoNC scenario.}  \label{fig11}
\end{figure}

\subsection{EE Gains from Dynamic Switching between MDNC and NoNC}
By comparing $4$-relay ``NoNC+PA" and $4$-relay  ``MDNC+PA" curves in  Fig. \ref{fig7}, we can see that  MDNC outperforms NoNC. We can achieve the same conclusion by comparing between $3$-relay ``NoNC+PA" and $3$-relay  ``MDNC+PA" scenarios. That is, with the same number of relays being selected, MDNC can lead to EE gains. { Additionally, when RS is implemented, we can obtain that MDNC can outperform the NoNC scheme in some target outage probability region, where the numbers of selected relays for MDNC and NoNC are the same or very close.}   {Specifically, in the region $\Pr_{out}\in (5\times10^{-4}, 3\times10^{-3})$,  $2$ and $3$ relays are respectively needed for ``NoNC+PA+RS"  and ``MDNC+PA+RS".    As can be seen, ``MDNC+PA+RS" achieves a higher EE, which  implies the EE gains resulting from network coding.  The reason is given as follows. We note that the total consumed energy consists of two parts, the electrical circuit energy and the data transmission energy.  Though the MDNC scheme consumes  more energy  in \emph{electrical devices and circuits} due to one extra activated relay (i.e., $R_2$), less \emph{data transmitting energy} is consumed as shown in Fig.~\ref{fig8} and Fig.~\ref{fig11}, such that the total consumed energy is less than that of the NoNC scenario. Thus, higher EE is obtained in the MDNC scenario.} On the other hand, in other outage probability regions, NoNC can perform better than the MDNC scheme if much fewer relays are needed for the NoNC scenario. In the region $\Pr_{out}\in (6\times10^{-7}, 5\times10^{-6})$ where $2$ and $4$ relays are respectively needed for ``NoNC+PA+RS"  and ``MDNC+PA+RS", ``NoNC+PA+RS" outperforms ``MDNC+PA+RS". This is because  electrical devices consume more energy to support another $2$ relays in the network coding scenario. Thus, a dynamic policy at the relays, i.e., switching between NoNC and MDNC, based on the target outage probability can potentially improve the EE.


\subsection{EE Gains from PA}
To clearly show the energy saving effect from PA, we plot the curves of $Etot_{data}$ versus $\Pr_{out}$ for both $3$ and $4$-relay scenarios in Fig. \ref{MDNC3relays4relays}. The $Etot_{data}$  gap between the curves for $4$-relay ``MDNC+PA" and $4$-relay  ``MDNC" (MDNC without PA or RS) in Fig. \ref{MDNC3relays4relays}  shows the energy saved with PA. Similar conclusions can also be obtained from $3$-relay ``MDNC+PA" and $3$-relay  ``MDNC". It is clear that for the same $\Pr_{out}$, more than $30\%$ of energy is saved. 
\begin{figure}
\centering
\includegraphics[width=0.56\textwidth]{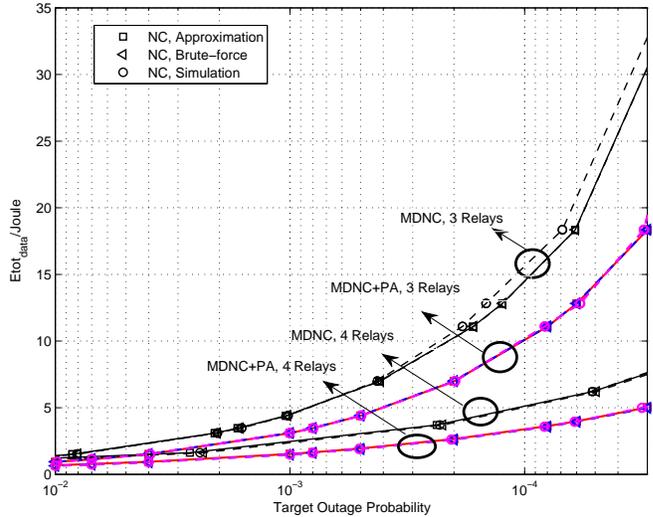}
\caption{$Etot_{data}$ versus the outage probability  for MDNC  scenario.}  \label{MDNC3relays4relays}
\end{figure}

\subsection{Effect of Circuit Energy Costs on RS}

From Fig. \ref{MDNC3relays4relays}, we can  conclude that if the circuit energy costs are ignored, the EE of four relays is larger than that of three relays. This is due to the fact that with more relays assisting in forwarding messages, a higher diversity order can be  provided. Contrarily, if we take the  electrical consumption into account, more relays may result in a decreasing EE,  as shown in Fig. \ref{fig7}. This is because that the outage probability gains achieved by adding extra relays are at the cost of basic electrical  circuits consumption.  Thus, we cannot ignore the circuit energy consumption. Based on that knowledge, the relay scheduling, including the modes (active, sleeping or off) and the number of relays, are significant.

\subsection{Effects of the Relay Locations}
{To investigate the effect of the relay locations on the EE, we move the relays  farther away from/closer to the users and closer to/farther away from the BS.} To be specific, the distance between $U_i$, $\forall i$,  and $R_j$, $\forall j$, is changed into $({d_{{ij}}}+\Delta)$ meters while the distance between $R_j$ and the BS becomes $({d_{j}}-\Delta)$ meters, where $\Delta$ can be regarded as the shifting distance of one relay.
\begin{figure}
\centering
\includegraphics[width=0.55\textwidth]{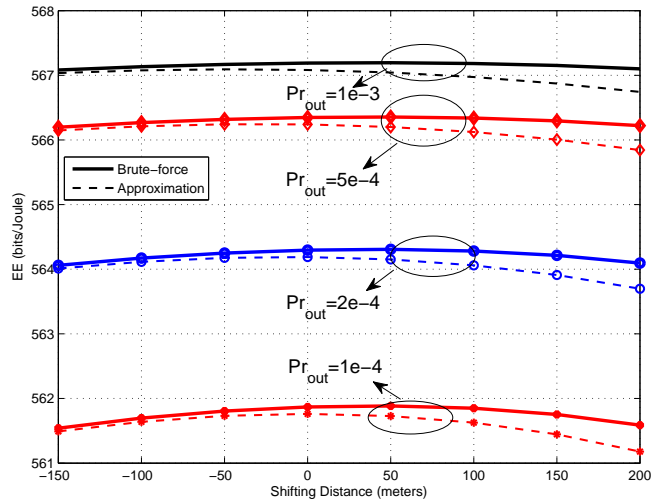}
\caption{3-Relay ``MDNC+PA"  scenario; shifting distance versus EE.}  \label{relaylocationeffect}
\end{figure}
As can be seen from Table II, $3$ relays are selected to optimize the EE when $1\times 10^{-5}\leq \Pr_{out, target} \leq 10^{-3}$.  In Fig. \ref{relaylocationeffect}, we draw the EE curves versus the shifting distance for $3$-Relay ``MDNC+PA" scenario when $\Pr_{out}=10^{-3}$, $5\times 10^{-4}$, $2\times 10^{-4}$ and $1\times 10^{-4}$, respectively. It can be found that the optimum value of  $\Delta$ with respect to the EE is $50$. For $50\leq \Delta \leq 200$, increasing $\Delta$ causes a larger path loss in the source-relay channels. Note that the first hop transmission is dominant in the two-hop transmission scheme. The  increasing path loss in  the first hop deteriorates the outage performance, which needs more energies in the second hop to compensate. On the other hand, for $-150\leq \Delta \leq 50$, an decreasing $\Delta$ yields a decline in the EE. This is because when the relays get closer to the destination node, the path loss of the relay-BS channel becomes larger. In this case, larger transmitting power is needed at relays and EE correspondingly decreases.


\section{Conclusions}

The energy efficiency of MDNC-based MUMR networks without CSIT has been studied. To formulate the EE maximization problem under constraints of the outage probability,  the exact outage probability expression was first derived but shown to be not tractable mathematically. Motivated by this, the outage probability was tightly approximated by its log-convex form.  The method  can be extended into various networks, such as  NoNC, TWR and ANC, \emph{etc}. Then we eventually transformed our original problem into  convex mixed-integer nonlinear one, for which a generalized outer approximation (GOA) algorithm was applied to decompose RS and PA such that they can be solved in an iterative manner. {In particular, PA and RS were respectively converted into a  standard convex and mixed-integer linear problem and  efficiently solved with interior point method and branch-and-cut algorithm}. Analytical results showed that the proposed PA and RS strategies can improve the EE. Further, when the number of selected relays was the same for MDNC and NoNC, it turned out that MDNC outperformed NoNC in terms of the EE. However, with RS, the NoNC scenarios may outperform the MDNC scenario in some outage probability region, where fewer selected relays were needed for the NoNC scenario and thus the circuit energy costs were lowered. Additionally, we have shown that transmitting power with more relays can be decreased, since with more relays assisting in forwarding messages, larger diversity was obtained. Moreover, the optimum relay positions with respecting to the EE were found at a  specific outage probability. It was demonstrated that a further increase/decrease in the transmission distance of the first/second hop both deteriorated the EE performance.


\section*{Appendix}

\subsection{Proof of Lemma \ref{Theorem1}}\label{lemmaproof}
\begin{proof}
In \cite{nonl}, the following conclusion has been obtained.
Let $N({x})$ and $D({x})$ be continuous and real-valued functions. Suppose $D({x})>0$. Then it is proved that
\begin{align} {q^ *} = N({x^ *})/D({x^ *}) = \max \{ N(x)/D(x)\},\end{align}

if and only if,

\begin{align}
F( q^ *) =F( q^ *, x^ *)= \max \{N(x) - { q^ * }D(x)\} = 0. \label{non}
\end{align}

{We note that if $\mathcal{U}$ is fixed, both $M{\alpha_0}(1-{\Pr}_{out})$ and ${E_{tot}}$ are continuous and real-valued functions and ${E_{tot}}>0$.  For $\mathbf{P1}$, by replacing $x^ *$, $N({x})$ and $D(x)$ in \eqref{non} with $[{{\mathbf{\tilde P}}_{\mathbf{s}}}^*,{{\mathbf{\tilde P^*}}_{\mathbf{MDNC}}}]$, $M{\alpha_0}(1-{\Pr}_{out})$ and ${E_{tot}}$, respectively, then Lemma $1$ can be obtained.} \quad \end{proof}

\subsection{Proof of Proposition \ref{prob1}}
It is obvious that  \eqref{eq63} and \eqref{eq631a} are convex. With fixed ${\mathcal U}$, \eqref{eq27}, \eqref{eq222} and $E_{tot}$ are the sum of multiple exponential terms multiplied by positive constants, which are in the geometric programming form. Then their log-form i.e., $\tilde V$ and \eqref{poutfinal} are convex \cite{24}. Thus, $\mathbf{P2}$ is convex. The efficient interior-point method \cite{24} can be applied to obtain the solution.\qed


\balance

\end{document}